\documentclass[11pt]{article}
\usepackage{a4wide}
\usepackage{amsmath}
\usepackage{amssymb}
\usepackage{amsthm}
\usepackage{graphicx}
\usepackage{color}
\usepackage{float}

\floatstyle{ruled}
\newfloat{algorithm}{htbp}{loa}
\floatname{algorithm}{Algorithm}

\theoremstyle{plain} \newtheorem{lemma}{Lemma}
\theoremstyle{plain} \newtheorem{theorem}{Theorem}
\theoremstyle{plain} \newtheorem{corollary}{Corollary}
\theoremstyle{plain} \newtheorem{remark}{Remark}
\theoremstyle{plain} \newtheorem{definition}{Definition}

\title{\textbf{Bi-Criteria and Approximation Algorithms for Restricted Matchings}\thanks{A preliminary version of this work appeared in Proceedings of the \textit{2nd International Symposium on Combinatorial Optimization} (ISCO 2012), Lecture Notes in Computer Science.   Supported by the Swiss National Science Foundation Project N.200020-122110/1 ``Approximation Algorithms for Machine Scheduling Through Theory and Experiments III" and by Hasler Foundation Grant 11099.}}

\author{Monaldo Mastrolilli and Georgios Stamoulis \\ Istituto Dalle Molle di Studi sull'Intelligenza Artificiale \\ IDSIA (USI/SUPSI) \\ Manno-Lugano\\ Switzerland \\ \textsf{\{monaldo, georgios\}@idsia.ch}}

\begin{document}

\date{}
\maketitle

\begin{abstract}
In this work we study approximation algorithms for the \textit{Bounded Color Matching} problem (a.k.a. Restricted Matching problem) which is defined as follows: given a graph in which each edge $e$ has a color $c_e$ and a profit $p_e \in \mathbb{Q}^+$, we want to compute a maximum (cardinality or profit) matching in which no more than $w_j \in \mathbb{Z}^+$ edges of color $c_j$ are present. This kind of problems, beside the theoretical interest on its own right, emerges in multi-fiber optical networking systems, where we interpret each unique wavelength  that can travel through the fiber as a color class and we would like to establish communication between pairs of systems. We study approximation and bi-criteria algorithms for this problem which are based  on linear programming techniques and, in particular, on polyhedral characterizations of the natural linear formulation of the problem. In our setting, we allow violations of the bounds $w_j$ and we model our problem as a bi-criteria problem: we have two objectives to optimize namely (a) to maximize the profit (maximum matching) while (b) minimizing the violation of the color bounds. We prove how we can ``beat" the integrality gap of the natural linear programming formulation of the problem by allowing only a slight violation of the color bounds. In particular, our main result is \textit{constant} approximation bounds for both criteria of the corresponding bi-criteria optimization problem.

\end{abstract}

\section{Introduction}											                                               	 

Consider the following game: we organize a competition in a school and we have a set of \textit{binary} games such as chess, GO, tavli (a.k.a.backgammon) etc. Some pairs of students are interested in playing one particular game, whereas some other pairs are interested in some other game. We only have  a limited amount of free boards for a particular game. Ideally, we would like to satisfy as many pairs of students as possible with the available amount of boards. This simple  game captures exactly the essence of the problem of this article: we can formulate the above scenario as a graph $G=(V,E)$, where $V$ is the set of students, and two students that are interested in a particular game (say chess) are connected with an edge of a particular color (say black) associated with this game. Let $w_{j} \in \mathbb{N}$ be the number of available boards of game $j$. Then the task of the organizers is to compute a maximum matching that uses at most $w_j$ boards of game $j$. Call this problem  Bounded Color Matching problem.

More formally, the problem can be stated as follows:

\begin{definition}[Bounded Color Matching]
We are given a (simple, un-directed) graph $G=(V,E)$ with vertex set $V$ and edge set $E$ such that $|V|=n$ and $|E|=m$. The edge set is partitioned into $k$ sets $E_1 \cup E_2 \cup \cdots \cup E_k$ i.e. every edge $e$ has color $C_j$ if $e \in E_j$ and a profit $p_e \in \mathbb{Q}^+$. We are asked to find a maximum (weighted) matching $M$ (or a matching of maximum cardinality) such that in $M$ there are no more that $w_j$ edges of color $C_j$, where $w_j \in \mathbb{Z^+}$ i.e. a matching $M$ such that $|M \cap E_j| \leq w_j$, $\forall j \in [k]$.
\end{definition}

In the following, we denote as $\mathcal{C}$ the collection of all the color classes. In other words, $\mathcal{C}= \{ C_j \}_{j\in [k]}$. Moreover, for a given edge $e \in E(G)$, we denote by $c(e)$ its color i.e. $c(e) = C_j \Leftrightarrow e \in E_j$.

\medskip
Besides the previously mentioned toy problem, Bounded Color Matching emerges in optical networking systems: in an optical fiber we allow multiplexing of different frequencies (i.e. different beams of light can travel at the same time inside the same fiber), but we have limited capacities of the number of light beams of a particular frequency that we allow to travel simultaneously through the system, due to potential interference problems. We would like to establish connections between a maximum number of (disjoint) pairs of systems while at the same time respecting the maximum number of connections using the same frequency we allow in multiplexing. Moreover, the Bounded Color Matching problem with 2 colors (i.e. the case that each edge is colored either blue or red) can be used in approximately solving the Directed Maximum Routing and Wavelength Assignment problem (DirMRWA) \cite{conf/infocom/NomikosPZ03} in \textit{rings} which are fundamental network topologies, see \cite{DBLP:conf/mfcs/NomikosPZ07} (also \cite{DBLP:conf/stacs/Caragiannis07, DBLP:journals/siamdm/Caragiannis09} for alternative and slightly better approximation algorithms and  \cite{DBLP:journals/networks/BampasPP11} for combinatorial algorithms). Here, approximately solving means that an (asymptotic) $\alpha$-approximation algorithm for maximum blue-red matching results in an (asymptotic) $\frac{\alpha +1}{\alpha+2}$-approximation algorithm for DirMRWA in rings.

\subsection{History and Related Results}
Characterizations and algorithms for maximum matchings in graphs have a very long history. One of the first attempts to characterize the structure of matchings was as early as the 1957 when Claude Berge characterized the structure of maximum matchings with respect to alternating and augmenting paths \cite{berge}: a matching $M$ on a given graph $G$ is maximum if and only if $G$ contains no $M$-augmenting paths. A path that alternates between edges in $M$ and edges not in $M$ (for a given matching $M$) is called an $M$-alternating path. An $M$-alternating path whose endpoints are unsaturated  by $M$ (i.e. vertices that do not have edges incident to them that are in $M$) is called an $M$-augmenting path.  $M$-augmenting paths provide a certificate of the non-maximality of $M$.

Given this characterization of maximum matchings, an algorithm is immediate for computing a maximum matching $M$ on a graph $G$:

\begin{quote}
\textsf{Initialize} $M := \emptyset$

\textsf{while} there exists an $M$ augmenting path $P$

\textsf{do} augment $M$ along $P$.
\end{quote}

Of course the running time of the above general algorithm depends on how fast we can find $M$-augmenting paths on a graph $G$ with $m$ edges and $n$ vertices. In case the graph is bipartite finding maximum matchings  can be done (relatively) easily in time $\mathcal{O}(m \sqrt{n})$ \cite{DBLP:conf/focs/HopcroftK71, DBLP:journals/siamcomp/HopcroftK73}, beating the ``trivial" brute-force approach which simply enumerates all possible $M$ augmenting paths which takes time $\mathcal{O}(nm)$. We remind that a graph $G$ is bipartite if its vertex set can be partitioned into two sets $V_1, V_2$ such that  every edge connects a vertex in $V_1$ to one in $V_2$; that is, $V_1$ and $V_2$ are independent sets. Equivalently, a bipartite graph is a graph that does not contain any odd-length cycles.

The case when $G$ is not bipartite is significantly more complicated because of the presence of odd-length cycles. In his seminal 1965 article, Jack Edmonds presented a $\mathcal{O}(n^2m)$ time algorithm for solving the maximum matching problem in general graphs \cite{Edmonds1965a}. In fact, it was precisely this article that introduced the concept of polynomially time  solvable problems as ``tractable" problems. As it always happens, the running time of this algorithm has been significantly improved over the years. In \cite{DBLP:conf/focs/EvenK75}, by a sophisticated use of some data structures,  a running time of $\min\{\sqrt{n}m \log n, n^{2.5}\}$ was shown for  computing a maximum matching in a graph $G$, which was later improved to $\mathcal{O}(n^{2.5})$, see \cite{DBLP:conf/focs/MicaliV80}.

The previous algorithms are purely combinatorial. Other very successful approaches for computing maximum matchings in graphs are based on using of algebraic methods and/or randomization. We will not go into much detail here, except mentioning the most important results, which include a $\mathcal{O}(n^{\omega+1})$ time algorithm \cite{DBLP:journals/jal/RabinV89}, and two $\mathcal{O}(n^{\omega})$ time algorithms \cite{DBLP:conf/focs/MuchaS04} and \cite{DBLP:conf/focs/Harvey06} (see also \cite{DBLP:journals/siamcomp/Harvey09}), where $\omega = \inf\{ c:$ two $n \times n$ matrices can be multiplied in time $\mathcal{O}(n^{c})  \}$.

All the above algorithms work for the \textit{unweighted} (uniform weights) case. In case we have a weight function $p: E \rightarrow \mathbb{Q}^{+}$ and we want to compute a maximum weighted matching, then other techniques are required. The most common technique is the so called \textit{Hungarian method} \cite{Kuhn1955}, which is  a primal-dual technique, initially introduced for bipartite graphs. For general graphs, similar primal-dual techniques have been employed, see for example \cite{Edmonds1965b}, \cite{DBLP:conf/focs/Gabow83, DBLP:journals/jcss/Gabow85} and \cite{DBLP:journals/jacm/GabowT91} among others. The idea, as most of the primal-dual schemata, is to build up feasible primal and dual solutions simultaneously and show that at the end both solutions satisfy complementary slackness conditions and hence by the duality theorem, the primal solution is a maximum weight matching. Another approach for the maximum weight matching problem is to maintain a feasible matching and try to successively augment it to increase its weight, until no more augmenting is possible, see for example \cite{cunningham_marsh}. For comprehensive accounts of the matching problem, we refer to \cite{matching_plummer} and \cite{polyhedra}.

\subsection{Constrained Matching Problems}
Since the task of computing maximum matchings is an extremely well studied and basic problem,  the interest has shifted towards some other versions of maximum matchings, in particular to versions where we seek a maximum matching subject to additional criteria (constraints). These criteria reduce the feasible solution space, making it (usually) harder to compute optimal solutions in polynomial time. In this direction,  Bounded Color Matching problems were studied as early as the 1970s as a very interesting generalization of the classical maximum matching problem: In \cite{garey_johnson}, the problem was defined as \textit{Multiple Choice Matching} (reference problem [GT55]) and proved to be \textbf{NP}-complete even for the very special case where the graph is bipartite, each color class contains at most 2 edges (i.e. $|E_j| \leq 2 $, $\forall j$) and $w_j = 1$, $\forall C_j \in \mathcal{C}$.  This problem, finds numerous practical applications, from classroom scheduling to image segmentation among others, see also \cite{DBLP:conf/icalp/ItaiR77}, \cite{DBLP:journals/jacm/ItaiRT78}.

Moreover, the uniform weight version of Bounded Color Matching problem is also closely related to the \textit{Labeled Matching} problem \cite{DBLP:conf/isaac/Monnot05},  \cite{DBLP:journals/cor/CarrabsCG09} in which all bounds $w_j$ are set equal to 1, i.e. we would like a maximum matching with at most one edge per color. In  \cite{DBLP:conf/isaac/Monnot05} it was proven that even the very special case of 2-regular bipartite graphs where each color appears twice (i.e. in at most two edges), the problem is \textbf{APX}-hard and so a PTAS is immediately out of reach for Bounded Color Matching (see also \cite{DBLP:journals/ipl/Monnot05}).

Budgeted versions of the maximum matching problem have been recently studied intensively. Here, by budgeted version of a combinatorial optimization problem $\Pi$ we mean the following: Besides the profit function $p: \mathcal{F} \rightarrow \mathbb{Q}^+$ associated with every feasible solution $F \in \mathcal{F}$ for $\Pi$ (where $\mathcal{F}$ is the set of all \textit{feasible} solutions for $\Pi$), we are also given a set of $\ell$ \textit{cost} functions $\{ \varrho_i \}_{i \in [\ell]}$ such that $\varrho_i : \mathcal{F} \rightarrow \mathbb{Q^+}$ , and for every cost function $\varrho_i$ a \textit{budget} $\beta_i \in \mathbb{Q}^+$. The budgeted optimization version of $\Pi$, which we call $\Pi_b$,   can be then formulated as follows (assuming that $\Pi$ is a maximization problem):

\begin{equation}
\max ~~p(F), ~~~\textrm{subject to} ~~~ F \in \mathcal{F} ~~\textrm{and}~~ \rho_i(F) \leq \beta_i,~~\forall i \in [\ell]
\end{equation}

In \cite{DBLP:conf/esa/GrandoniZ10} (see also \cite{DBLP:journals/corr/abs-1002-2147}) the authors considered the 2-budgeted maximum matching problem (i.e. the case where $\ell = 2$) and devised a PTAS.  This algorithm works roughly as follows: First of all, a guessing step is performed that guesses the $1/\epsilon$ most valuable edges of the optimal matching. Then, an optimal \textit{fractional} matching $x^*$ is computed for the rest of the graph (for example by solving the linear programming relaxation of the problem). By Caratheodory's theorem \cite{caratheodory}, $x^*$ can be written as a \textit{convex} combination of at most three (possibly unfeasible) matchings i.e. $x^* = \lambda_1 x_1 + \lambda_2 x_2 + (1-\lambda_1 - \lambda_2) x_3$. Then, the algorithm consists of two ``merging" steps: in the first step, given the first two matchings $x_1$ and $x_2$ the output is a third matching $z$ with comparable profit and which is not costlier than $\mu x_1 + (1-\mu)x_2$ for $\mu = \frac{\lambda_1}{\lambda_1 + \lambda_2}$ with respect to both the two extra cost functions. Then the same procedure is again applied to $z$ and $x_3$ with parameter $\mu=\frac{\lambda_1 + \lambda_2 }{ \lambda_1 + \lambda_2 + (1 - \lambda_1 -\lambda_2)}$, i.e. we merge $x_3$ with $z$ such that the new matching $z^*$ is feasible (with respect to both cost functions) and (almost) optimal.

This was further improved in \cite{DBLP:conf/soda/ChekuriVZ11} where it is given a PTAS for a \textit{fixed} number of  budgets. The authors there provided both randomized and deterministic PTAS's for the problem. The randomized version is based on strong concentration bounds of some suitable martingale processes (see also \cite{DBLP:conf/focs/ChekuriVZ10} for some closely related results and techniques). The deterministic PTAS can be seen as a \textit{bi-criteria} approximation, and the final solution returned is within $(1-\epsilon)$ the optimal but it might violate the budgets by a factor of $(1+\epsilon)$ (i.e. the solution $z$ returned has the property that $c_i(z) \leq (1+\epsilon) \beta_i$, $\forall i$). Moreover, for unbounded number of budgets the authors prove an almost optimal approximation guarantee, but with allowing a very large (i.e. \textit{logarithmic}) overflow on the budgets (as before, this means that for the computed solution $z$, $z$ has the property that $c_i(z) \leq \beta_i \log \beta_i$, $\forall$ budget $i$). These results generalize the results for the budgeted \textit{bipartite} matching problem, for which a PTAS was known for the case of one budget \cite{DBLP:conf/ipco/BergerBGS08, DBLP:journals/mp/BergerBGS11}, or in the case of fixed number of budgets \cite{DBLP:conf/esa/GrandoniRS09} in which a $(1-\epsilon, 1+\epsilon)$ bi-criteria approximation was shown.

To the best of our knowledge, the first case where matching problems with cardinality (disjoint) budgets were considered, was in \cite{DBLP:conf/mfcs/NomikosPZ07} where the authors defined and studied the \textit{blue-red} Matching problem:  compute a maximum (cardinality) matching that has at most $w$ blue and at most $w$ red edges, in a blue-red colored (multi)-graph. A $\frac{3}{4}$ polynomial time combinatorial approximation algorithm and an $\mathbf{RNC^2}$ algorithm were presented (that computes the maximum matching that respects both budget bounds with high probability). We note that the exact complexity of the blue-red matching problem is not known: it is only known that blue-red matching is at least as hard as the Exact Matching problem \cite{DBLP:journals/jacm/PapadimitriouY82} whose complexity is open for more than 30 years. A polynomial time algorithm for the blue-red matching problem will imply that Exact Matching is polynomial time solvable. On the other hand, blue-red matching is probably not \textbf{NP}-hard since it admits an $\mathbf{RNC^2}$ algorithm. We note that this algorithm can be extended to a constant number of color classes with arbitrary bounds $w_j$. Using the results of \cite{DBLP:journals/algorithmica/Yuster12}  (also appeared in \cite{DBLP:conf/approx/Yuster07}) one can deduce an ``almost" optimal deterministic algorithm for blue-red matching, i.e. an algorithm that returns a matching of maximum cardinality that violates the two color bounds by at most one edge. This is the best possible, unless of course blue-red matching (and, consequently, exact matching) are in \textbf{P}.

\subsection{Our Contributions}
In this article we study the Bounded Color Matching problem, from a Linear Programming point of view. In particular, we are interested how good approximation algorithms we can design using linear programming methods. The main contribution of the current manuscript is to show how we can ``beat" the integrality gap of the natural LP formulation of the BCM problem, allowing small violation of the color bounds $w_j$.

Before we do that, we firstly  prove that a simple greedy and fast procedure gives a $\frac{1}{3}$ approximate solution. To prove the approximation guarantee of this simple procedure, we use a characterization that was introduced in \cite{menstre_greedy}  to show that our problem falls into the framework of $\ell$-extendible systems. This serves the purpose of a baseline and ``warm-up" result.

Then we design and analyze various algorithms based on Linear Programming techniques. Our algorithms are based on iterative rounding  of basic (fractional) feasible solutions  of the natural Linear Programming formulation of the Bounded Color Matching problem. We employ a fractional charging technique (introduced in \cite{DBLP:journals/siamcomp/BansalKN09}) to characterize the structure of extreme point solutions of the LP relaxation of our problem. Taking advantage of this structure, we provide bi-criteria additive and multiplicative approximation algorithms for both the weighted and unweighted case (see \cite{mohit_thesis} and also \cite{iterative_book} for a comprehensive account of the applications of iterative rounding techniques in the context of combinatorial optimization).

Very generally, our algorithms have two (global) steps:

\begin{enumerate}
\item[-] Either (iteratively) apply a \textit{rounding} step on some variable with high fractional value in such a way that the resulting solution remains feasible, or
\item[-] apply a \textit{relaxation} step in which we decide to drop a budget constraint if a constraint with ``few" non-zero variables exists.
\end{enumerate}

\noindent
Our results (and the structure of this document) can be summarized as follows:

\begin{enumerate}
\item Firstly, as already mentioned, we show that a straightforward greedy strategy results in an $\frac{1}{3}$-approximation guarantee.
\item In the next section we prove some combinatorial properties of the natural linear programming formulation and we apply these techniques in the special case of the BCM problem where $w_j = 1, \forall j \in [k]$. We note that this case remains \textbf{APX}-hard (see related work section). We provide an \textit{asymptotic} approximation of the optimal objective function value by allowing a small \textit{additive} violation of the color bounds $w_j$. In particular we prove that there exists a polynomial time algorithm that, for any $\alpha \in \mathbb{Z}^+$ (in fact we require that $\alpha$ is greater than 3 on bipartite and greater than 4 in general graphs), it computes a matching of value at least $opt (1 - \frac{4}{\alpha}) + \frac{1}{\alpha} +1$ that has at most $\alpha$ edges of every color (where $opt$ is the optimal solution value). This result can be improved to $opt (1 - \frac{3}{\alpha}) + \frac{1}{\alpha} +1$ with the same additive $\alpha - 1$ violation bound in case the graph is bipartite. This means that, by allowing a moderate \textit{additive} violation of $\alpha -1$ for every color class, we can approximate the optimal objective function value within any precision (that depends on $\alpha$). The result holds for the uniform weight case.
\item Then, we further investigate some properties of the underlying polyhedron and using these properties we show how we can in fact obtain a family of bi-criteria approximation algorithms with varying guarantees. In particular, we prove the following:
    \begin{enumerate}
    \item There exists an $\frac{1}{2}$-approximation algorithm for the weighted version of the BCM problem, allowing only an additive one violation of the color bounds $w_j$.
    \item We prove that, for any $\lambda \in [0,1]$, there is a polynomial time $(\frac{2}{3+\lambda}, \frac{2}{1+\lambda} + \frac{1}{w_j})$ bi-criteria approximation algorithm for the un-weighted Bounded Color Matching problem i.e. we prove \textit{constant} approximation bounds with respect to \textit{both} criteria. We note that, to the best of our knowledge, this is the first result that provides such performance guarantees (compare with the $(1-\epsilon)$ approximation but with logarithmic  budget violation of \cite{DBLP:conf/soda/ChekuriVZ11}).
    \item Finally, we present a polynomial time $\frac{1}{2}$-approximation algorithm for the uniform weight Bounded Color Matching problem \textit{without} any violation of the color budgets, matching the integrality gap of the natural linear relaxation of the problem.
     \end{enumerate}
\end{enumerate}

\section{Preliminaries and Theoretical Framework}\label{preliminaries}

Consider the classical matching problem on a general graph $G=(V,E)$. If we introduce binary variables $x_e,~\forall e = \{u,v\} \in E(G)$ where $x_e = 1 \Leftrightarrow e \in M$, then we can describe the problem of finding a maximum matching with a linear program as follows: find a vector $x \in \{ 0,1 \}^{|E|}$ that maximizes $1^T x$ (or $p^T x$ for a general profit vector $p \in \mathbb{Q}^{|E|}_+$) such that adjacent to each vertex, there is at most one variable (edge) that takes the value one. In particular, the linear program is the following:


\begin{eqnarray}\label{bipartite_polyhedron}
\Bigg\{ \max p^T x, ~s.t.~  \sum_{e \in \delta(v)} x_e \leq 1,~\forall v \in V, ~~ x \in \{0,1\}^{|E|}   \Bigg\}
\end{eqnarray}

\noindent
where $\delta(v) = \{ e \in E(G): v \in e  \}$, $\forall v \in V(G)$. The constraint of the form $\sum_{e \in \delta(v)} x_e \leq 1$, $\forall v \in V$, simply tells us that we seek a solution that has at most one edge incident to every vertex of the graph. Unfortunately, solving this integer program is \textbf{NP}-hard. The usual thing to do it to relax the integrality constraints of the variables, i.e., replace the constraint $x \in \{0,1\}^{|E|}$  with the $x \in [0,1]^{|E|}$. This relaxed program can be solved in polynomial time \cite{schrijver}. The problem is that, in general, solving the relaxation of an integer program results in a \textit{fractional} vector. In some cases, we are able to add some valid constraints that cause the solution to be always integer, and this is the case with the maximum matching problem in general graphs: the problem is caused by structures called \textit{blossoms} i.e. holes (cycles) of odd cardinality. Let $S$ be such a cycle of cardinality $|S|$, which is an odd number. Any maximum matching can contain at most $\frac{|S|-1}{2}$ edges from this cycle, but the linear programming relaxation can assign value $\frac{1}{2}$ to every edge of the cycle for a fractional vector of value $\frac{|S|}{2} > \frac{|S|-1}{2}$. To solve this, we add the so-called \textit{blossom inequalities} that constraint exactly this: in every odd cycle, we require that the sum of the values assigned to its edges is at most $\frac{|S|-1}{2}$ and this result to the following restricted polyhedron: Find $x \in [0,1]^{|E|}$ that maximizes $p^Tx$ such that

\begin{eqnarray}\label{general_polyhedron}
\Bigg\{ \sum_{e \in \delta(v)} x_e \leq 1,~\forall v \in V,~~~\sum_{e \in E(S)} x_e \leq \frac{|S| -1}{2},~  \forall S \subseteq V, ~~|S| ~~\textrm{odd cardinality}    \Bigg\}
\end{eqnarray}

\noindent
where by $E(S)$ for $S \subseteq V$ we denote the set of edges included in the graph induced by the vertex set $S$. Although the above linear program described in (\ref{general_polyhedron}) has \textit{exponential} number of constraints for general graphs (one for every vertex and one for every odd sized subset of vertices), we can still solve it in polynomial time by the Ellipsoid method if we provide a separation oracle, which for any given candidate solution  vector ${x} \in [0,1]^{|E|} $  will either respond that ${x}$ is a feasible solution for the linear programming inequalities defined above, or, it will respond that ${x}$ is infeasible by providing the violated constraint. A very interesting fact is that by solving this linear relaxation, the resulting vector is always integral!


These ``blossom" constraints are redundant in case of bipartite graphs (since in bipartite graphs every cycle is of even length), but are essential in our general graph setting. So, when we consider bipartite graphs we will assume that these constraints will not be part of the LP and we will be just using the  initial degree-constrained polyhedron described by (\ref{bipartite_polyhedron}). Again here we have the phenomenon that by solving the linear relaxation, the resulting vector is again integral \cite{polyhedra, combinatorial_optimization, cunningham_marsh, matching_plummer}.


Call the polyhedron that contains all feasible points of (\ref{general_polyhedron}) or (\ref{bipartite_polyhedron}) as $\mathcal{M}$. We use the same name for both polyhedra, but from the context will be clear which one we are using, thus avoiding any confusion. For a comprehensive treatment of the various properties of $\mathcal{M}$, including a polynomial time separation oracle, we refer to \cite{matching_plummer}.



We can describe the set of all feasible solution of the constrained (Bounded Color) matching problem as follows:

\begin{eqnarray}
\label{bcm_polyhedron}
\mathcal{M}_c = \Bigg\{ y \in \{0,1\}^{E} :~~ y \in \mathcal{M} ~~ \bigwedge ~~ \sum_{e \in E_j} y_e \leq w_j, ~~\forall j  \in [k]    \Bigg\}
\end{eqnarray}

To judge the quality of a linear programming relaxation, and to explore the limits of linear programming techniques, the concept of \textit{integrality gap} (or integrality \textit{ratio}) has been introduced: informally a linear relaxation is strong when it does not allow a lot of ``cheating" with regard the objective function value over the original integral formulation. This is called \textit{integrality gap} and is defined as follows: Let us assume that we have an optimization problem with a set of valid instances $I$ and let $\mathcal{Z}$ the set of all feasible solutions for a particular instance $\in I$ defined by the integer formulation of the problem which we call it IP. Define $\textsf{opt}(IP) = \textsf{opt}(\mathcal{Z}) = \max_{z \in \mathcal{Z}} f(z)$ where $f(z)$ is the objective function value of the feasible solution $z$. As before, let $\mathcal{Z}'$ be the set of feasible points on the linear relaxation of IP (which we call it LP) and define $\textsf{opt}(LP)\textsf{opt}(\mathcal{Z}') = \max_{q \in \mathcal{Z}'} f(q)$. Then the integrality gap (or the \textit{integrality ratio}) of the relaxation of IP is

\begin{displaymath}
\sup_{i \in I} ~\frac{\textsf{opt}(LP)}{\textsf{opt}(IP)}
\end{displaymath}

Of course, the closer this quantity is to one, the better the quality of the formulation. An LP formulation with integrality gap of $\varrho$ implies that it is impossible to design an approximation algorithm with performance guarantee better than $\varrho$ using this particular formulation as upper/lower bounding schema for our discrete optimization problem.

Figure 1 shows that the integrality ratio of $\mathcal{M}_c$ in both the general and the bipartite case is at most 2 already for 3 color classes with bounds $w_j = 1$, thus we cannot hope to achieve a better than $\frac{1}{2}$ approximation using the natural LP relaxation of the Bounded Color Matching problem as defined by $\mathcal{M}_c$. The main contribution of this manuscript is to show that if we compromise a little and allow violation of the color constraints, we can design algorithms with better than $\frac{1}{2}$ approximation guarantee, thus ``beating" the natural barrier caused by the integrality ratio of the $\mathcal{M}_c$ polyhedron.

\begin{figure}[h!]
\label{example}
\centering
\includegraphics[scale=0.70]{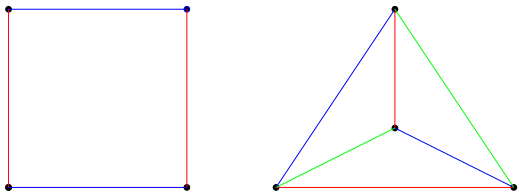}
\caption{In both cases we seek a maximum matching with at most 1 edge per color. Any optimal \textit{integral} solution has value 1 in both graphs. Observe that in the first case the all $\frac{1}{2}$ solution gives a solution of value 2 but every integral optimal solution has value 1. The same is true in the second graph which is not bipartite: by assigning the value $\frac{1}{3}$ to all the edges we get an optimal solution of value 2 (note that the all $\frac{1}{3}$ solution is \textit{not} a basic feasible solution). }
\end{figure}

Since we are interested in bi-criteria approximation algorithms we need to define formally what we mean by that:

\begin{definition}
An algorithm $\mathcal{A}$ for the Bounded Color Matching problem is an $((\alpha,\beta),(\gamma,\delta))$ bi-criteria approximation algorithm, for $\alpha, \beta,\gamma,\delta \in \mathbb{Q}^{\geq 0}$, if

\begin{enumerate}
\item $\mathcal{A}$ runs in polynomial time in the input size of the instance,
\item $\mathcal{A}$ returns a solution $sol$ for which $sol \geq \alpha \cdot opt + \beta$, where $opt$ is (an upper bound on) the optimal solution value and,
\item the solution that $\mathcal{A}$ returns has at most $\gamma \cdot w_j + \delta$ edges for every color class $C_j, j \in [k]$.
\end{enumerate}
\end{definition}

In the following lemma, we show that using linear programming techniques, it is impossible to design bi-criteria approximation algorithms with certain performance guarantees, even for very simple cases. In particular, we show  that it is impossible to obtain an additive error with the natural LP formulation  without violating the objective function value.

\begin{lemma}
It is impossible to design bi-criteria approximation algorithms for the BCM problem of the form $((1,0), (0, \zeta))$, for any additive value of $\zeta$ using the natural linear programming formulation of the problem.
\end{lemma}

\begin{proof}Consider the following instance of the BCM problem: we are given a (bipartite) graph $G = (V,U,E)$ with bipartition $U,V$ on $2n$ vertices ($|V | = |U| = n$) which is actually a path i.e. $E = {(v_i, u_i) \cup (v_{i+1}, u_i)}_{i=1,...n}$. We have only one budget of the form $\sum_{e \in R} x_e \leq \frac{n}{2}$ where $R = {(v_i, u_i)}_{i=1,...n}$. We solve the linear program (\ref{bcm_polyhedron}) to obtain an optimal basic solution $x$. This basic feasible solution corresponds to the vector $x = (1/2, . . . , 1/2)^T$ . So, in this example, the number of fractional edges for the cardinality constraint is $n$ which is twice as much as the the bound of the constraint which proves that there are not always constraints that have œ``few" non-zero variables. On the other hand, observe that the optimal (fractional) solution has value of $\frac{2n-1}{2}$ since we have $2n-1$ edges (a path on $2n$ vertices) and each edge get value $\frac{1}{2}$. On the other hand, even after dropping the budget constraint by the edges defined by $R$, the optimal (integral) solution has value $\frac{n-2}{2}$ and so we cannot ``reach" the optimal fractional solution. In other words, we cannot hope to achieve any additive violation on the budget constraint, for any constant value, without violating the objective function.
\end{proof}

And so, we have to ``violate" the objective function value (i.e. settle for an approximate solution) if we wish to achieve a constant violation on the color bounds. Please relate this impossibility result with the result of \cite{DBLP:journals/algorithmica/Yuster12} where it is shown how to compute, via combinatorial methods, maximum matchings with at most an additive one violation of the color budgets, for two color classes.

\section{A Fast Greedy $1/3$ Approximation Algorithm}

In this section we consider the weighted variant of the bounded color matching problem where each edge $e$ has a profit $p_e \in \mathbb{Q}^+$. The goal  is to find a maximum profit matching that respects the color bounds.   Here we show how a $\mathcal{O}(m \log m)$ greedy procedure can easily derive a $\frac{1}{3}$ approximation for general weighted graphs.

\begin{algorithm}
\label{Greedy}
\caption{Greedy algorithm for Bounded Color Matching problem}

\noindent
\textit{\textbf{Input:}} Graph $G=(V,E)$, a color function $c: e \rightarrow [k]$, a profit function $p: e \rightarrow \mathbb{Q}^+$.

\noindent
\textit{\textbf{Output:}}A matching $M$ such that $|M \cap E_j| \leq w_j$, $\forall$ color class $j \in [k]$.

\begin{enumerate}
\item \textbf{initialize} $M := \emptyset$.
\item Sort all edges of the graph according to their profits in non-increasing order.
\item \textbf{while $E(G) \neq \emptyset$} \textbf{do}:
	\begin{itemize}
	\item[-] Pick the edge $e$ with the largest profit.
	\item[-] $M := M \cup \{ e\}$. $w_{c(e)} := w_{c(e)} - 1$. $E(G) := E(G) \setminus \{e'\in E: e\cap e'\neq \emptyset\}$.
	\item[-] \textbf{if} $w_{c(e)} = 0$ \textbf{then} remove all edges of the same color from the graph.
	\end{itemize}
\item \textbf{return} $M$.
\end{enumerate}
\end{algorithm}

To analyze the performance guarantee of the above simple procedure, we will use the notion of $\ell$-extendible systems due to  Mestre \cite{menstre_greedy}:

\begin{definition}
A \textbf{subset system} is a pair $(\mathcal{E}, \mathcal{L})$, where $\mathcal{E}$ is a finite ground set of elements and $\mathcal{L} \subseteq 2^{\mathcal{E}}$ with the property that $L \in \mathcal{L} \Rightarrow L' \in \mathcal{L}$, $\forall L' \subset L$. We say that $L_1 \in \mathcal{L}$ is an \textbf{extension} of $L_2 \in \mathcal{L}$ if $L_2 \subseteq L_1$.

A subset system $(\mathcal{E}, \mathcal{L})$ is said to be \textbf{$\ell$-extendible} if $ ~\forall ~X \in \mathcal{L}$, $x \notin X$ with $X \cup \{ x \} \in \mathcal{L}$ and for every extension $Y$ of $X$, $\exists ~Y' \subseteq Y \setminus X$ with $|Y'| \leq \ell$ such that $Y \setminus Y' \cup \{ x \} \in \mathcal{L}$.
\end{definition}

What is helpful, is the fact that the greedy algorithm applied to $\ell$-extendible systems provides a $\frac{1}{\ell}$ factor approximation. The following result is from \cite{menstre_greedy}:

\begin{theorem}
Let $(\mathcal{E}, \mathcal{L})$ be an $\ell$-extendible system. Then the Greedy algorithm applied to such systems provides an $\frac{1}{\ell}$ approximation algorithm for the optimization problem for any weight (or profit) function $p$.
\end{theorem}

\begin{lemma}
The subset system associated with the Bounded Color Matching problem is 3-extendible.
\end{lemma}

\begin{proof}
Let $M$ and $M'$ be feasible solutions such that $M'$ is an extension of $M$ (i.e. $M \subseteq M'$). Let $M$ be such that $M \cup \{ e \}$ is still feasible, for some edge $e = \{ u,v \} \in E$ with color $j$ such that $e \notin M'$.
By the above property that $M \cup \{ e \}$ is feasible, it is easy to see that $\deg_M (u) = \deg_M (v) = 0$ and $|M \cap E_j| < w_j$. Now consider the extension of $M$, namely $M'$. We can find at most 3 edges $ e_1, e_2, e_3 \in M'$ such that $u \in e_1$, $v \in e_2$ and $c(e_3) = c(e)$. Observe that we can find many edges $e_3$. But \textit{any} such edge would suffice. The point is that the addition of $e$ in $M'$ would potentially lead to at most three ``conflicting" edges (the addition of $e$ would cause the removal of at most three edges in order the new solution to remain feasible).  Consider the new solution $Z = M' \setminus \{ e_1, e_2, e_3 \} \cup e$. This is still a feasible solution for the Bounded Color Matching problem with $|\{e_1, e_2, e_3  \}| \leq 3 = \ell$, and so the system characterizing our problem is 3-extendible. Observe that we use inequality in the $|\{e_1, e_2, e_3  \}| \leq 3$ because it might be, for example, that $e_1 = e_3$.
\end{proof}

\begin{corollary}
Algorithm 1 is an $\mathcal{O}(m\log m)$ time $\frac{1}{3}$-approximation algorithm for the weighted version of the Bounded Color Matching problem in general graphs.
\end{corollary}

\noindent
Figure 2 shows that this bound is essentially tight.

\begin{figure}[h!]\label{example_greedy}
\centering
\includegraphics[scale=0.80]{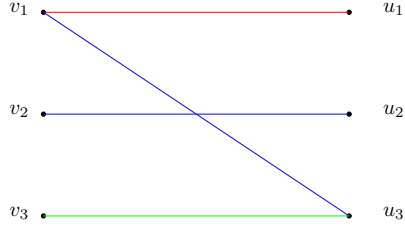}
\caption{An example of worst case behavior of the Greedy procedure. Assume that all the color bounds are set to 1. Edges $\{v_1,u_3\}$ and $\{v_2,u_2\}$ are blue while $\{v_1,u_1\}$ is red and $\{v_3,u_3\}$ is green. Selecting the (blue) edge $(v_1,u_3)$ will result to a solution of value 1. On the other hand the optimal solution consists of selecting the edges $(u_1,v_1), (u_2,v_2), (u_3,v_3)$.}
\end{figure}

\section{Combinatorial Properties of the $\mathcal{M}_c$ Polyhedron}

In this subsection we will  prove some  interesting combinatorial properties of extreme point solutions of the polyhedron $\mathcal{M}_c$. Later on, we will take advantage of these properties to devise our approximation algorithms.  We first need some preliminary definitions.

\begin{definition}
Let $E' \subseteq E$ be a subset of the edges of the graph. Then, we define the \textbf{characteristic vector} of $E'$  to be the binary vector $\chi_{E'} \in \{ 0,1 \}^{E}$ such that $\chi_{E'}(e) = 1 \Leftrightarrow e \in E'$ i.e. the $i$-th component of $\chi_{E'}$ is 1, if the $i$-th edge belongs to $E'$ and zero otherwise.
\end{definition}

\begin{definition}
Let $y$ be a real-valued vector in an $n$-dimensional space. Define the \textbf{support} of $y$ to be the indices of all the non-zero components of $y$ i.e. $support(y) = \{ i \in [n]: y_i \neq 0  \}$.
\end{definition}

\begin{definition}
A family $\mathcal{L}$ of subsets of some universe $U$ is called \textbf{laminar} if it is not intersecting i.e. for any two subsets $L_1,L_2 \in \mathcal{L}$ either $L_1 \subseteq L_2$, or $L_2 \subseteq L_1$, or $L_1 \cap L_2 = \emptyset$.
\end{definition}

Now, if we solve (to optimality) the relaxation of the linear program defined by $\mathcal{M}_c$, we will obtain a basic feasible solution\footnote{We note that basic feasible solutions and extreme point solutions are equivalent concepts} (in fact, an \textit{optimal} basic solution) $x^*$. We can characterize this basic solution $x^*$ as follows (for general graphs):

\begin{lemma}\label{rank_lemma}
Let $x^*$ be an optimal basic solution for the LP relaxation (relaxing the constraints $x_e \in \{ 0,1 \}$ with $x_e \in [0,1]$) of (\ref{bcm_polyhedron}) (with the blossom inequalities) such that $x^*_e > 0$ $\forall e \in E$. Then, there exist $\mathcal{F}$ $\subseteq V$, a family $\mathcal{L}$ $\subseteq 2^V$ of odd cardinality subsets of vertices and $\mathcal{Q}$ $\subseteq [k]$ such that

\begin{enumerate}
\item $\sum_{e \in \delta(v)} x^*_e = 1$, $\forall v \in \mathcal{F}$.
\item $\sum_{e \in E(S)} x^*_e = \frac{|S|-1}{2}$, $\forall S \in \mathcal{L}$.
\item $\sum_{e \in E_j}x^*_e = w_j$, $\forall j \in \mathcal{Q}$.
\item $\{\chi_{\delta(v)}\}_{v \in \mathcal{F}}$, $\{\chi_{E(S)}\}_{S \in \mathcal{L}}$ and $\{ \chi_{E_j} \}_{j\in \mathcal{Q}}$ are all linearly independent, i.e. the linear constraints corresponding to $\mathcal{F}$, $\mathcal{L}$ and $\mathcal{Q}$ are al linearly independent.
\item $|E| = |\mathcal{F}|+ |\mathcal{L}| + |\mathcal{Q}|$ i.e. the number of edges (non-zero variables) is equal to the number of tight, linearly independent constraints.
\end{enumerate}
\end{lemma}

The lemma follows by basic properties of the basic feasible solutions:

\begin{theorem}[\cite{schrijver}]\label{tight}
Let $P = \{ x~|~ Ax \leq b   \}$ where $x \in \mathbb{R}^n, A \in \mathbb{R}^{m \times n}$ and $b \in \mathbb{R}^m$ be a polyhedron in $\mathbb{R}^n$. Then, a point $z \in \mathbb{R}^n$ is a vertex of the polyhedron $P$ if and only if $\textrm{rank}(A_{|z}) = n$ where $A_{|z}$ is the sub matrix of $A$ consisting of those rows $i$ such that $A_i z = b_i$.
\end{theorem}

Indeed, we can form a basic feasible solution by selecting $|E|$ linearly independent constraints from our linear program, set them to equality, and solve the linear system. The last item in the lemma simply says that the number of non-zero variables (which corresponds to edges in the residual graph) is simply the number of linear independent constraints set to equality when we obtain the linear system. The assumption that $x^*_e >0$ implies that all constraints that we set to equality must come from the first three types of constraints (vertex, blossom and color constraints), but not non-negativity constraints.

From now on, when we refer to a \textit{tight vertex} $v$ we will mean a vertex such that the constraint corresponding to that vertex is tight i.e. $\sum_{e \in \delta(v)} x_e = 1$. Similar for \textit{tight color class} (i.e. a color class $j$ such that $\sum_{e \in E_j}x_e = w_j$) and tight odd-cardinality vertex sets (i.e. subsets $S$ of the vertices of odd cardinality such that $\sum_{e \in E(S)} x_e = \frac{|S|-1}{2}$). Observe that, in general, not all tight vertices belong to $\mathcal{F}$ and not all tight colors belong to $\mathcal{Q}$ (the same for all tight odd cardinality subsets of vertices). But every element of $\mathcal{F} \cup \mathcal{Q} \cup \mathcal{L}$ is tight.

\medskip
A most important result concerning the family $\mathcal{L}$ of odd cardinality subsets of vertices is that it can be taken to be \textit{laminar} (non-intersecting).

\begin{lemma}[\cite{polyhedra},\cite{combinatorial_optimization},\cite{cunningham_marsh}, \cite{matching_plummer}]
The family $\mathcal{L}$ of odd cardinality subset of the vertex set as defined above can be taken to be laminar.
\end{lemma}

The proof of the above argument, uses standard uncrossing techniques and can be found in the above mentioned references. A useful observation, that we will need shortly, is about the cardinality of $\mathcal{L}$,  $|\mathcal{L}|$:

\begin{lemma} \label{laminar}
Let $\mathcal{L}$ be a laminar collection of odd cardinality subsets of a universe of elements $U$, such that $|L| \geq 3$, $\forall L \in \mathcal{L}$. Then $|\mathcal{L}| \leq \lfloor \frac{|U|-1}{2} \rfloor$.
\end{lemma}

\begin{proof}
We will prove the above statement by induction (on the size of the universe of elements $U$). For the base case we have that if $|U| = 3$ then $\mathcal{L}$ can have only one odd cardinality subset of the elements of $U$ with cardinality at least 3: namely, the whole set $U$ and so $|\mathcal{L}| = 1$ and the statement is trivially true.

Now suppose that $|U| > 3$. Let $\mathcal{L}$ be a laminar family of $U$ satisfying the conditions of the lemma with the maximum possible cardinality (among all other candidates laminar families of $U$). In this $\mathcal{L}$, let $S$ be a subset with the smallest cardinality i.e. $S = \arg \min_{L \in \mathcal{L}} |L|$. Trivially, $S$ has cardinality 3. Now, we remove from $U$ all elements $u \in S$ except one. Let $\bar{U}$ be the new universe with all but one of the elements of $S$ removed from $U$. Observe that all sets $L \in \mathcal{L}$, $L \neq S$ still fulfill the conditions of the lemma. Let $\bar{\mathcal{L}}$ be the new laminar family on $\bar{U}$. By the induction hypothesis  $|\bar{\mathcal{L}}| \leq \frac{|\bar{U}|-1}{2}$ and, by the previous observation,  $|\bar{\mathcal{L}}| = |\mathcal{L}| - 1$ ( $\bar{\mathcal{L}}$ is just $\mathcal{L}$ simply without $S$, and the rest of the sets are present). So, we have that

\begin{eqnarray*}
|\bar{\mathcal{L}}| \,\,\,\leq\,\,\, \frac{|\bar{U}| -1}{2} \,\,\,=\,\,\, \frac{|U| -2 -1}{2} \,\,\,=\,\,\, \frac{|U| - 1}{2} -1
\end{eqnarray*}

\noindent
and so we conclude that

\begin{eqnarray*}
 |\mathcal{L}| \,\,\,=\,\,\, |\bar{\mathcal{L}}| + 1  \,\,\,\leq\,\,\, \frac{|U| - 1}{2} -1 +1 \,\,\,=\,\,\, \frac{|U|-1}{2} \text{.}
\end{eqnarray*}
\end{proof}

We would like an upper bound on the total number of the tight linearly independent constraints that constitute $x^*$ in terms of $opt = \sum_e x^*_e$ i.e. the optimal (fractional) solution value of the relaxation of the natural LP for the Bounded Color Matching problem.

\begin{lemma}\label{corollary}
Take any basic feasible solution $x^* \in (0,1]^{|E|}$. Then, for the sets $\mathcal{F}, \mathcal{L}, \mathcal{Q}$ characterizing the solution $x^*$ as described in Lemma \ref{rank_lemma} we have that

\begin{equation}
|support(x^*)| =  |\mathcal{F}| + |\mathcal{L}| + |\mathcal{Q}| \leq 4opt
\end{equation}

\noindent
for general graphs and

\begin{equation}
|support(x^*)| = |\mathcal{F}| + |\mathcal{Q}|    < 3opt
\end{equation}

\noindent
for bipartite graphs.
\end{lemma}

\begin{proof}
As usual, let  $opt$ be the  optimal solution value of the relaxation, i.e., $opt = \sum_{e} x^*_e$. Given $opt$ we would like to enumerate how many constraints we can have from each family of tight, linearly independent constraints $\mathcal{F}$, $\mathcal{Q}$ and (in the case of general graphs) $\mathcal{L}$ that characterize $x^*$.

First of all, it should be clear  that $|\mathcal{Q}| \leq opt$, i.e. the number of tight \textit{color} constraints in $\mathcal{Q}$ is at most $opt$ (since $w_j \in \mathbb{Z}^+$). In general, if we denote by $\xi = \min_{j \in [k]} w_j$, then $|\mathcal{Q}| \leq \frac{\mathrm{opt}}{\xi}$.

Second, consider the maximum (cardinality) matching on a graph on $|V|$ vertices. If $|V|$ is even we can have at most $\frac{|V|}{2}$ edges (actually this is the case of the perfect matching), otherwise we can have at most $\lfloor \frac{|V|}{2} \rfloor$ edges. We will show that $|\mathcal{F}| \leq 2opt$. In the case the graph $G$ is bipartite, this is easy: suppose that the graph is bipartite and  that $|\mathcal{F}| > 2opt$. Then at least one side of the bipartition must have strictly more than $opt$ tight vertices and if we sum the value of the edges incident to these vertices we would get value greater than $opt$, a contradiction (remember that a vertex $v$ is tight if $\sum_{e \in \delta(v)} x^*_e = 1$).  Now assume that the graph is not bipartite. Again, it is not hard to show that $|\mathcal{F}| \leq 2opt$. Indeed,

\begin{displaymath}
|\mathcal{F}| \,\,\,\le\,\,\, |V| \,\,\,\le\,\,\, \sum_{v \in V} \sum_{e \in  \delta(v)} x_e^* \,\,\,=\,\,\, 2 opt.
\end{displaymath}

Finally, from Lemma \ref{laminar} we have that $\mathcal{L} \leq \lfloor \frac{|V| -1}{2} \rfloor$. This means that we can have at most $\lfloor \frac{|V|- 1}{2} \rfloor$ tight inequalities from the corresponding set of constraints defined by $\mathcal{L}$. Observe that the largest family $L \in \mathcal{L}$ can have cardinality at most $2opt+1$: assume that this is false and there is $S \in \mathcal{L}:$ $|S| > 2opt +1$. Then we have that

\begin{eqnarray*}
\sum_{e \in E(S)} x^*_e = \frac{|S|-1}{2} \geq \frac{2opt +2 -1 }{2} > opt
\end{eqnarray*}

\noindent
a contradiction. So $\max_{S \in \mathcal{L}} |S| \leq 2opt+1$. By the laminarity of $\mathcal{L}$, and since we consider only odd subsets, an immediate application of Lemma \ref{laminar} on the maximal subsets $S \in \mathcal{L}$ gives us the desired bound that $|\mathcal{L}| \leq opt$ and this completes the proof of the lemma.
\end{proof}

\subsection{An Application of Lemma \ref{corollary}}

As a first step, we consider the special case where all $w_j$'s are equal to 1, i.e., we want to find a maximum cardinality matching that has at most one edge from each color. We will prove that if we allow a moderate additive violation of the color budgets, we can approximate the optimal objective function value within any desired accuracy.

We consider the following algorithm (see Algorithm \ref{assymptotic}). We require that the parameter $\alpha$ is greater or equal than 3 for bipartite instances or greater or equal than 4 for general graph instances.

\begin{algorithm}
\caption{Algorithm for 1-Bounded Color Matching (with parameter $\alpha$)}
\label{assymptotic}

\noindent
\textit{\textbf{Input:}} An un-weighted graph $G=(V,E)$, a color function $c: E \rightarrow [k]$,

parameter $\alpha \in \mathbb{Z}^+$ such that $\alpha \geq 3$ if $G$ is bipartite, else $\alpha \geq 4$.

\noindent
\textit{\textbf{Output:}}A matching $M$ such that $|M \cap E_j| \leq \alpha$, $\forall$ color classes $j \in [k]$.

\begin{enumerate}
\item \textbf{initialize} $M := \emptyset$.
\item Solve the Linear Programming relaxation of the (current) problem to obtain an optimal basic solution $x$:
    \begin{itemize}
    \item[-] \textbf{if} $G$ is bipartite solve (\ref{bcm_polyhedron}) with (\ref{bipartite_polyhedron}) as $\mathcal{M}$,
    \item[-] \textbf{else} use (\ref{bcm_polyhedron}) with (\ref{general_polyhedron}) as $\mathcal{M}$.
    \end{itemize}
     Define $E(G) = \{x_i: i \in support(x)\}$.
\item \textbf{for every} edge $e \in E$ such that $x_e = 1$ \textbf{do}
	\begin{enumerate}
	\item[$\circ$] Add $e$ to $M$.
	\item[$\circ$] Delete $e$ and the endpoints of $e$ from $G$,

    remove the constraint for color class $C_j$ such that $e \in E_j$,

    remove the constraints of the vertices $\{u,v\} = e$,

    continue (\textbf{goto} step 2).
	\item[$\circ$] \textbf{if} $|support(x)| = 0$ \textbf{then} \textbf{return} $M$.
	\end{enumerate}
\item \textbf{Relaxation:} \textbf{if} there is a tight color class $C_j$ such that $|E_j| \leq \alpha$

		\textbf{then} relax the constraint for this color class.
\item \textbf{Rounding:} \textbf{else} there is an edge $e$ belonging to a tight color class $C_j$ such that $x_e < 1/ \alpha$
	\begin{itemize}
	\item[$\circ$] Round $x_e$ to zero and continue (i.e., remove $e$ from the current solution) (\textbf{goto} step 2).
	\end{itemize}
\end{enumerate}
\end{algorithm}

Observe that in case step 4 is not performed then this means that all tight color classes have support greater than $\alpha$ and the corresponding variables sum up to one (since $w_j =1, \forall j \in [k]$). So we have strictly more than $\alpha$ variables summing up to one and so we conclude that at least one should be less than $\frac{1}{\alpha}$ . Let the linear program in the $\psi$-th iteration be $\text{LP}_{\psi}$. Define the distance between the value of the LP between two iterations $\psi$ and $\psi+1$ of the Algorithm 1 to be $\Delta_{\psi} = value(\text{LP}_{\psi}) - value(\text{LP}_{\psi+1})$ where $value(\text{LP})$ is the (optimal) solution value of the linear program LP. It should be clear from step 5 of the algorithm that $0 \leq \Delta_{\psi} \leq 1/\alpha$ for every iteration.

\begin{lemma}
Algorithm \ref{assymptotic} is a $((1-\frac{3}{\alpha} + \epsilon,0),(1, \alpha-1))$ bi-criteria approximation algorithm for the uniform weight 1-Bounded Color Bipartite Matching problem.
\end{lemma}

\begin{proof}
First of all, it is easy to see that the algorithm terminates in polynomial time and the solution returned is indeed a matching. In fact, we can have at most $|E|$ rounding steps (where $|E| = support(x)$ is the number of edges of the \textit{initial} graph) and at most $|E|$ relaxation steps (one for each color).

The fact that the algorithm returns a solution that violates every color constraint by at most an additive $\alpha-1$ comes from the relaxation step: we relax the constraint of a color class only when $|E_j \cap support(x^*)| \leq \alpha$, so, in the worst case, we will include all these edges in our final solution resulting in a surplus of at most $\alpha-1$ edges.

To see that the algorithm returns a $(1-\frac{3}{\alpha})$ approximate solution on the objective function, we notice that in each step the value of the LP solution decreases by at most $1/\alpha$ i.e. between two consecutive iterations $i$ and $i+1$ in which we perform a rounding step in the $i$-th we have that $\Delta_{i} = value(\text{LP}_{i}) - value(\text{LP}_{i+1}) \leq \frac{1}{\alpha} $ and, moreover, one edge is deleted from the graph. By Lemma \ref{corollary}, the number of edges is at most $ 3 opt - 1$. So, we can have at most $3opt -1 -\alpha$ iterations (because when we have fewer than $\alpha$ edges, clearly we can perform a relaxation step) and in each iteration we lose at most $1/\alpha$ for a total loss of

\begin{displaymath}
\frac{1}{\alpha} \cdot (3opt -1 -\alpha) = opt  \frac{3}{\alpha} - \frac{1}{\alpha} -1
\end{displaymath}

\noindent
and so

\begin{eqnarray*}
sol \geq opt - \frac{3opt}{\alpha} + \frac{1}{\alpha} +1 = opt \Big(1 - \frac{3}{\alpha}\Big) + \frac{1}{\alpha} +1
\end{eqnarray*}

\noindent
 from which we conclude that

\begin{eqnarray}
\frac{sol}{opt} \geq \frac{opt(1 - \frac{3}{\alpha}) + \frac{1}{\alpha} +1}{opt}   = \Big( 1 - \frac{3}{\alpha} \Big) + \epsilon
\end{eqnarray}

\noindent
for $\epsilon =  \frac{1}{\alpha \cdot opt} + \frac{1}{opt}$.  The proof for general graphs is identical, by just replacing the $3opt - 1$ number of iterations with the upper bound of edges for general graphs (due to Lemma \ref{corollary}).
\end{proof}

\begin{corollary}
There exist  polynomial time $((1-\frac{3}{\alpha} + \epsilon,0),(1, \alpha-1))$ and $((1-\frac{4}{\alpha} + \epsilon,0),(1, \alpha-1))$ bi-criteria approximation algorithms for the Bipartite and General Graph 1-Bounded Color Matching problem respectively.
\end{corollary}

\section{A Characterization of  Basic Feasible Solutions of the $\mathcal{M}_c$ Polyhedron}

In this section we  will prove that basic feasible solutions of the LP relaxation of our problem ($\mathcal{M}_c$) have certain properties that will allow us to design better and more general approximation algorithms. In some sense, we prove that every extreme point solution $x$ must be ``sparse" (i.e. its support size is relatively small). By taking advantage of the sparsity of such solutions, we will  design approximation algorithms that ``beat" the $\frac{1}{2}$ integrality gap (modulo a slight violation of the budget constraints). Our algorithms  are based on the iterative rounding approach \cite{mohit_thesis} (see also \cite{iterative_book} for a comprehensive account of applications of iterative methods in combinatorial optimization). We employ a fractional charging technique (which was first introduced in \cite{DBLP:journals/siamcomp/BansalKN09}) to characterize the structure of extreme point solutions of the LP relaxation of our problem.

Recall Lemma \ref{rank_lemma}. This lemma characterizes all basic feasible solutions $x \in (0,1]^{|E|}$: every basic feasible solution must respect  Lemma \ref{rank_lemma}. But we can make some additional observations regarding the structure of any basic feasible solution $x$ (recall that the residual graph is the graph $G = (V, E)$ where $E = \{ e \in E: e \in support(x) \}$):

\begin{lemma}\label{lemma}
Let $x$ be any basic feasible solution $x$ such that $x_e > 0 ~\forall e$ (i.e. there is no edge with $x_e = 0$) in our LP relaxation $\mathcal{M}_c$ (without the blossom inequalities). Then one of the following must be true:
\begin{enumerate}
\item either there is an edge $e$ such that $x_e = 1$,
\item or there is a tight color class $j \in \mathcal{Q}$ such that $|E_j| \leq w_j+1$ in the residual graph,
\item or there is a tight vertex $v \in \mathcal{F}$ such that the degree of $v$ in the residual graph is 2.
\end{enumerate}
\end{lemma}

\begin{remark}
The above characterization is for the polyhedron defined by (\ref{bipartite_polyhedron}) plus the budget constraint inequalities for bipartite graphs. But, we may use these set of linear inequalities described by $\mathcal{M}_c$ (vertex degree constraints plus the color budget constraints) even for general graphs, without any loss of generallity, in contrast with the previous section where we used two different polyhedra for the bipartite and the general graph case. The reason is the following: in the algorithm of the previous section, we require that the final solution is integral after we drop the color budget constraints, so, for the general graph case, we need the characterization of (\ref{general_polyhedron}), as otherwise the final solution would not be integral and we would need an extra step to retrieve an integral solution. The problem lies on the integrality gap of the polyhedron (\ref{bipartite_polyhedron}) when the underlying graph is not bipartite, which is (essentially) $\frac{2}{3}$, and thus it would impossible to get arbitrary close to the optimal objective function value.

On the other hand,  the output of the algorithm of the next section is guaranteed to be integral and since the behavior of the two different formulations is essentially the same (they are both fractional polyhedra and they both have the same integrality gap), it is unnecessary to use the formulation defined by (\ref{general_polyhedron}). In other words, without any loss, we can use the simpler LP formulation described by (\ref{bipartite_polyhedron}) plus the extra linear color budget constraints even for the case that the graph is arbitrary.
\end{remark}

\begin{proof}
We will prove the claim of the lemma by deriving a contradiction. Assume that for all edges $e$ in the residual graph we have that $0 < x_e <1$. We will employ a fractional charging argument in which every edge $e$ with $x_e >0$ will distribute fractional charge to every tight object that is part of (which might be vertex or color class). We will employ the scheme in such a way that every edge gives a charge of at most 1, for a total charge of \textit{at most} $|E|$ (the number of edges in our residual graph). Then, we will show that every tight object will receive charge of \textit{at least} one, for a total collected charge of \textit{at least} $|E|$. In fact, we will show that the total charge distributed is \textit{strictly less} than $|E|$, deriving the desired contradiction. Our charging scheme will work based on the hypothesis of the lemma.

In fact, for the sake of contradiction, let us assume that in any basic feasible solution $x$ (such that $x_e \in (0,1) ~\forall e$) we have

\begin{enumerate}
\item for every tight color class $j \in \mathcal{Q}$, $|E_j| > w_j +1$ and
\item for every tight vertex $v \in \mathcal{F}:~\deg(v) \geq 3$.
\end{enumerate}

Now, consider the following charging scheme in which every (fractional) edge $e = (u,v)$, such that $e \in E_j$, distributes fractional charge as follows:

\begin{enumerate}
\item if $j \in \mathcal{Q}$, i.e. if the color of edge $e$ is tight, then $e$ distributes charge of $\frac{1}{2}(1-x_e) >0$ to the color class $C_j$.
\item every tight vertex $\{ u,v \} \in e$ that belongs to $\mathcal{F}$ receives from $e$ a charge of $\frac{1}{4}(1+x_e) <1$.
\end{enumerate}

\noindent
Observe that the total charge distributed by any edge is at most

\begin{eqnarray*}
\frac{1}{2}(1-x_e) + 2 \Big( \frac{1}{4}(1+x_e) \Big) = \frac{1 - x_e + 1 + x_e}{2} = 1
\end{eqnarray*}

\noindent
So, the total charge distributed by all (fractional) edges of the residual graph is \textit{at most} $|E|$.

Now, let us calculate the total charge received by every tight vertex $v \in F$ and every tight color class $C_j \in \mathcal{Q}$.  We first begin by the vertices $v \in \mathcal{F}$. Consider such a vertex. The total charge received by $v$ is the sum of the charges given to it by all edges incident to $v$:

\begin{eqnarray*}
\textrm{charge}(v)   =  \sum_{e \in \delta(v)} \frac{1}{4}(1+x_e) & = & \frac{1}{4} \sum_{e \in \delta(v)} (1+x_e) \\
                                                                  & = & \frac{1}{4} (|\delta(v)| + 1) \geq 1
\end{eqnarray*}

\noindent
the last inequality following by the hypothesis that all tight vertices $\in \mathcal{F}$ have degree at least 3. So, every tight vertex $v \in \mathcal{F}$ receives total charge of at least 1.

Now we calculate the total charge received by any tight color class $C_j \in \mathcal{Q}$. As before, the total charge received by any such color class is the sum of the charges given to $C_j$ by all fractional edges of color $j$:

\begin{eqnarray*}
\textrm{charge} (C_j)   =  \sum_{e \in E_j} \frac{1}{2}(1-x_e) & = & \frac{1}{2} \sum_{e \in E_j} (1-x_e) \\
                                                               & = & \frac{1}{2} (|E_j| - w_j) \geq 1
\end{eqnarray*}

\noindent
where in the last inequality we used the fact that $C_j \in \mathcal{Q} \Rightarrow |E_j| \geq w_j+2$ (by hypothesis). So, again we see that every tight color class $\in \mathcal{Q}$ receives charge of at least 1. We conclude that the total charge that has been distributed is \textit{at least} $|\mathcal{F}|+|\mathcal{Q}| = |E|$.

We need to calculate the total charge given by all (fractional) edges of the graph. We argued that the total charge given is \textit{at most} $|E| = |\mathcal{F}|+|\mathcal{Q}|$ since every edge distributes a charge of at most 1. But, we will show that the total charge given is \textit{strictly less} than $|E|$, giving us the desired contradiction.

Indeed, if for some edge $e = (u,v)$ belonging to color class $C_j$ we have that one of its endpoints $u$ or $v$ does not belong to $\mathcal{F}$, i.e. if $\{u,v \} \nsubseteq \mathcal{F}$, then a charge of $\frac{1}{4}(1+x_e) >0$ is wasted, so the total charge is strictly less than 1, which results to a total charge strictly less than $|E|$. Similarly, if $C_j \notin \mathcal{Q}$ then a charge of $\frac{1}{2}(1-x_e) >0$ is wasted, and again we have total charge less than $|E|$. So, we may assume that all vertices belong to $\mathcal{F}$ and all color classes belong to $\mathcal{Q}$. But then observe that

\begin{eqnarray*}
\frac{1}{2} \sum_{v \in V} \chi_{\delta(v)} \,\,\,=\,\,\, \sum_{C_j \in C} \chi_{E_j}
\end{eqnarray*}

\noindent
where $\chi_{\delta(v)} \in \{0,1\}^{|E|}$ is the characteristic vector of the edges whose one endpoint is $v$ (analogously for $\chi_{E_j}$). So, the characteristic vectors corresponding to the vertices are \textit{not} linearly independent, a contradiction. We conclude that in the absence of an edge with unit value, either there is a color class $C_j \in \mathcal{Q}:~|E_j| \leq w_j+1$ or a tight vertex $v \in \mathcal{F}:~\deg(v) = 2$.
\end{proof}


\begin{remark}
The statement of the lemma holds even when the $w_j$'s are fractional. In such a case we just replace the $w_j +1$ term on the claim of the lemma with $\lceil w_j \rceil +1$ and the lemma is still true.
\end{remark}

\subsection{A Simple Algorithm}										                                                    

Given Lemma \ref{lemma}, we propose the following simple algorithm for the weighted Bounded Color Matching problem (see Algorithm \ref{algo_weighted}). We solve the LP (the relaxation of the ILP defined in (1) by replacing the integrality bounds with $x_e \in[0,1],~\forall e$) and obtain a basic feasible solution $x$, we construct the graph $G'$ (which we call it residual graph) such that $G' = (V',E')$ where $V' = \{v \in V(G):~ \sum_{e \in \delta(v)} x_e >0  \}$ and $E' = \{e \in E(G):~x_e >0  \}$, and we either identify a color constraint to relax (relaxation step), or a vertex constraint to relax. We iterate until we have relaxed all constraints defined by $\mathcal{F}$ and $\mathcal{Q}$.

\begin{algorithm}
\caption{First algorithm for Bipartite bounded Color Matching.}
\label{algo_weighted}

\noindent
\textit{\textbf{Input:}} Graph $G=(V,E)$, a color function $c: e \rightarrow [k]$, a profit function $p: e \rightarrow \mathbb{Q}^+$. Bounds $w_j, \forall j \in [k]$.

\noindent
\textit{\textbf{Output:}}A graph $G'$ such that $|G' \cap E_j| \leq w_j+1$, $\forall$ color classes $j \in [k]$ and $\deg(v) \leq 2$, $\forall v \in V(G')$.


\smallskip
\textbf{initialize:} $M := \emptyset$

\medskip
\textbf{while} $C \neq \emptyset$ or $E \neq \emptyset$ \textsf{do}
	\begin{enumerate}
	\item[$\alpha.$] Compute an optimal (fractional) basic solution $x$ to the \textit{current} LP.
	\item[$\beta.$] Remove all edges from the graph such that $x_e = 0$.
	\item[$\gamma.$] Remove all vertices of the graph such that $\deg(v) = 0$.
	\item[$\delta.$] \textbf{if} $\exists e = (u,v) \in E:~x_e = 1$ and $e \in C_j$

    \textbf{then} $G' := G'\cup \{ e\}$, $V = V \setminus \{u,v\}$, $w_j := w_j - 1$.

    \textbf{if} $w_j = 0$

    \textbf{then} $C := C \setminus C_j$, $E := E \setminus \{e: e \in E_j\}$.
	\item[$\varepsilon.$] (\textbf{Relaxation:}) \textbf{while $V \cup \mathcal{C} \neq \emptyset$ }
	\begin{enumerate}	
	\item \textbf{if} $\exists$ color class $C_j \in Q$ with $|E_j| \leq w_j+1$
		
\textbf{then} remove the constraint for this color class i.e. define $\mathcal{C} := \mathcal{C} \setminus C_j$.
	\item \textbf{if} $\exists$ vertex $v \in F$ such that $\deg(v) = 2$

\textbf{then} remove the constraint for that vertex.
	\end{enumerate}
	\end{enumerate}

\textbf{return} $G'$
\end{algorithm}


In each step of the algorithm, either we drop a tight vertex constraint $v \in \mathcal{F}$, or we drop a tight color constraint for a color class $C_j \in \mathcal{Q}$. Thus the algorithm will terminate in at most $|\mathcal{Q}|+|\mathcal{F}|$ steps and in each step we need to resolve the current LP. Observe that at the end of the algorithm, the graph  $G'$ is a collection of disjoint paths or cycles: this is because we remove the degree constraints for a vertex $v$ only when $\deg(v) =2$, so every vertex in $G'$ will have degree at most 2 (because every vertex eventually will become tight), and so $G'$ is a collection of disjoint paths and cycles. Similarly, in $G'$ we can have at most $w_j +1$ edges for every color class.

Next, we use the following claim, which is immediate:

\begin{lemma}
The sum of the weight of the edges in $G'$ is at least $p^Tx$ where $x$ is the initial (optimal) basic feasible solution for the LP relaxation of the Bounded Color Matching problem.
\end{lemma}

Let $CC$ be the collection of all connected components of $G'$. Let $c \in CC$ be such a connected component. Because of the structure of $c$ we know that $c$ is a union of two (disjoint) matchings $M_1^c, M_2^c$ i.e. $M_1^c \cup M_2^c = c$. Now, let $x_c$ be the restriction of $x$ to the edges of $c$. We observe that one of the matchings $M_1^c$ or $M_2^c$ has weight at least $\frac{1}{2}p^Tx_c$. And this is true for every connected component $c \in CC$. So, for every component $c \in CC$ we include in $M$ that matching $M_i^c$, $i \in \{1,2\}$ such that $p(M_i^c) \geq \frac{1}{2}p^T x_c$. Since $p(G') = \sum_{e \in E(G')} p_e x_e \geq p^T x$ for the initial $x$, we have that $p(M) \geq \frac{1}{2} p^Tx$ and in $M$ we can violate every color constraint by at most an additive 1.

\begin{theorem}
There is a polynomial time $((1/2,0), (0,1))$ bi-criteria approximation algorithm for the weighted Bounded Color Matching problem.
\end{theorem}

\subsection{Bi-criteria Algorithms for the General Bounded Color Matchings}

We propose the following algorithm which is based on the observations made in Lemma \ref{lemma}. The idea, as explained in the end of the proof of the lemma, is the following: we solve the natural LP of the problem as described by the inequalities in $\mathcal{M}_c$ to obtain an optimal basic feasible solution $x$. If $\exists j \in \mathcal{Q}: |support(x) \cap \chi_{E_j}| \leq w_j + 1$ then we relax the constraint for the corresponding color i.e. we remove it from the set of inequalities of $\mathcal{M}_c$. If $\forall j \in [k], |support(x) \cap \chi_{E_j}| \geq w_j + 2$ then by Lemma \ref{lemma} we know that $\exists v \in \mathcal{F}: |support(x) \cap \chi_{\delta(v)}| \leq 2$ which implies that $\exists e \in \delta(v): x_e \geq \frac{1}{2}$. On this edge we perform a rounding step: we include this edge in $M$ (our solution), we remove all other ``conflicting" edges in order to have a feasible matching, we decrease the corresponding color bound of $e$ appropriately, we remove the constraints of the endpoints of $e$ from $\mathcal{M}_c$ and iterate. We give the details in Algorithm \ref{bicriteria}.

\begin{algorithm}
\caption{Bi-criteria Algorithm for the uniform weight Bounded Color Matching Problem}
\label{bicriteria}
\textit{\textbf{Input:}} A general graph $G = (V,E)$ such that each edge has a color $j \in [k]$. Color bounds $w_j \in \mathbb{Z}^+, \forall j \in [k]$.

\noindent
\textit{\textbf{Output:}} A matching $M$ such that $|M \cap E_j| \leq w_j$

\medskip
\noindent
\textbf{Initialize:} $M := \emptyset$.

\textbf{while} $\mathcal{Q} \neq \emptyset$ \textsf{do}:
\begin{enumerate}
\item Compute an optimal solution $x$ to the current LP relaxation of the problem (using (\ref{general_polyhedron}) as $\mathcal{M}$ in (\ref{bcm_polyhedron})).
\item \textbf{if} $\exists e \in E(G): x_e = 0$ \textbf{then} delete all such $e$ from the graph.
\item \textbf{if} $\exists i \in support(x): x_i =1$ \textbf{then}:
	\begin{enumerate}
	\item[-] $M := M \cup \{ i \}$.
	\item[-] \textbf{if} $i \in E_j$ for some $j \in [k]$ \textbf{then} set $w_j := w_j - 1$.
	\item[-] Delete $i$ and its endpoints from $G$.
	\end{enumerate}
\item (\textbf{Relaxation}:) \textbf{if} $\exists j \in \mathcal{Q}: |support(x) \cap \chi_{E_j}| \leq w_j +1$

\textbf{then} remove the constraint for $E_j$ from the current set of linear inequalities $\mathcal{M}_c$.

\item (\textbf{Rounding}:) \textbf{else} $\exists v \in \mathcal{F}: |support(x) \cap \chi_{\delta(v)}| \leq 2 $.

Let $u_1,u_2$ the neighbors of $v$ such that $e = \{ v,u_1 \}$. Let $e \in E_j$ for some $j \in [k]$.

\textbf{do}
	\begin{enumerate}
	\item[-] $M := M \cup e$.
	\item[-] $V = V \setminus \{ v,u_1 \}$.
	\item[-] $w_j := \max \{ 0, w_j - x_e - \lambda (1-x_e)$\}.
	\item[-] $support(x) := support(x) \setminus \Big\{ \{ v,u_1\}, \{v,u_2 \}, \bigcup_{z_i \in N(u_1)} \{ u_1, z_i \}   \Big\}$.
	\end{enumerate}
and iterate (i.e. \textbf{go to} step 1 with input the graph $G = (V,support(x))$ and the  updated bounds $w_j$).
\end{enumerate}

\textbf{return} $M$.
\end{algorithm}

Observe that when performing the rounding step, the natural choices are to decrease $w_j$ by $x_e$ (giving us better approximation bounds but with worse violation of the color constraints) or by 1 (giving us worse approximation bounds but with better i.e. additive 2 violations of the color bounds). Instead, we give the freedom to control the decrease by any intermediate value in $[x_e,1]$. In the following we will prove an exact bound on the trade off between approximation and violation based on the parameter $\lambda \in [0,1]$.  Observe that $x_e + \lambda(1-x_e) \in [x_e,1]$.

Define $\vartheta = 1 -x_e$, so the bound update step is $w_j := w_j - x_e - \lambda \cdot \vartheta$. First of all we observe that for $u_1$ as in the Rounding step of  Algorithm \ref{bicriteria} we have that

$$\sum_{e' \in \delta(u_1)\setminus \{e\}} x_{e'} \,\,\,\leq\,\,\, 1-x_e \,\,\,=\,\,\,  \vartheta$$.

\begin{lemma}
In each application of the Rounding step, the objective function decreases by at most $1 + \vartheta( 1+ \lambda)$.
\end{lemma}

\begin{proof}
Each Rounding step affects all the edges that are adjacent to $v$ and to $u$. Since $v \in \mathcal{F}$ we have that $\sum_{e \in \delta(v)} x_e = 1$. Moreover, $\sum_{e' \in \delta(u)\setminus \{e\}} x_{e'} \leq 1-x_e =  \vartheta$. So the loss due to rounding all the edges adjacent to $u$ and $v$ appropriately (as described in the corresponding step of the above algorithm) is at most $1 +\vartheta$. Besides this, a loss might occur because of the color bound update. This loss can be at most $\lambda \cdot \vartheta$. This is because  $x_e + \lambda \cdot \vartheta \geq x_e$, so we decrease $w_j$ by more than $x_e$ so that it could be the case $\sum_{e' \in E_j\setminus \{ e \}} x_{e'} > w_j$ for the new updated bound. But of course $\sum_{e' \in E_j \setminus \{ e \}} x_{e'} - w_j \leq \lambda \cdot \vartheta$ for the new updated bound $w_j$. So the value of the optimal bfs $x$ decreases by an additional factor of at most $\lambda \cdot \vartheta$ giving us a total decrease of at most $1+\vartheta(1+\lambda)$.
\end{proof}

See also Figure \ref{example_round} for an illustration of the Rounding Step.

\begin{figure}
\centering
\includegraphics[scale=0.80]{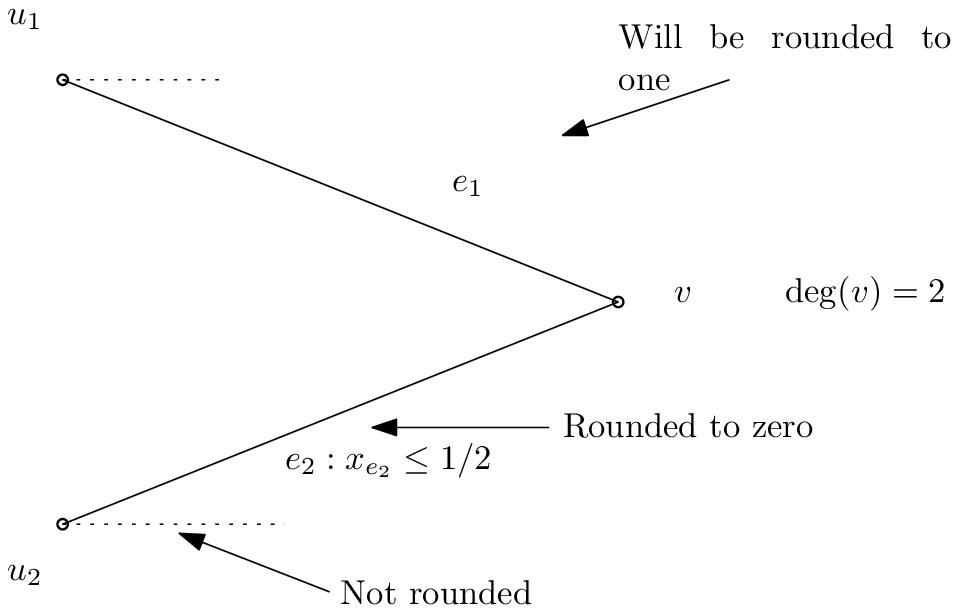}
\caption{The rounding case. We know that $v$ has degree 2. We round $e_1$ to 1, $e_2$ to zero and all edges adjacent to vertex $u_1$ will be rounded to zero. The edges adjacent to $u_2$ (besides $e_2$) will remain unchanged.}
\label{example_round}
\end{figure}

We see that in every application of a Rounding step, the objective function decreases by at most $1+\vartheta(1+\lambda)$ but we include one edge in our solution $M$, so the total true loss due to a single Rounding step is at most $\vartheta (1+\lambda)$. Since we can perform very few (namely at most $\lfloor \frac{|V|}{2} \rfloor$) iterations, we see that we can have few different values of $\vartheta$.

\begin{lemma}
Let $opt$ be the optimal objective function value and let $sol = |M|$, i.e., the value of the solution returned by Algorithm \ref{bicriteria}. Then, for every $\lambda \in [0,1]$, we have that $\frac{sol}{opt}\geq \frac{2}{2+(\lambda+1)}$.
\end{lemma}

\begin{proof}
Denote the value of $\vartheta$ in iteration $i$ as $\vartheta_i$ so in that iteration our objective function decreases by at most $1+\vartheta_i(1+\lambda)$ for a total loss of at most $\vartheta_i(1+\lambda)$. Observe that the maximum number of iterations we can perform for a particular value of $\vartheta_i$ before the optimal initial objective function value (which we denote by $opt$) truncates to zero is $\frac{opt}{1+\vartheta_i(1+\lambda)}$. Let $f_i$ be the fraction of the maximum number of times we can perform a rounding step with a particular value $\vartheta_i$, i.e.,

\begin{displaymath}
f_i = \frac{n_i}{opt} \cdot (1+\vartheta_i(1+\lambda))
\end{displaymath}

where $n_i$ is the number of times that a particular $\vartheta_i$ occurs. Observe that $\sum_i f_i \leq 1$. This is true since otherwise the total reduction in the objective function value would be $\sum_i f_i (1+\vartheta_i (1+\lambda))opt$ which would be strictly greater than $opt$.

We conclude that the final objective function value is

\begin{eqnarray*}
opt - \sum_i \frac{opt}{1 + \vartheta_i(\lambda+1)}\cdot f_i \cdot \vartheta_i(\lambda+1) = \mathcal{G}
\end{eqnarray*}

\noindent
where the term inside the summation corresponds to the total accumulated \textit{loss} occurring  for a particular value of $\vartheta_i$ (fraction of possible maximum number of iterations for this particular $\vartheta_i$ $\times$ actual loss).

Let $sol = |M|$ i.e. the size of the matching returned by the above Algorithm. By the previous discussion it is apparent that $|M| = sol = opt - \mathcal{G}$ and so

\begin{eqnarray*}
|M| = sol = opt -\mathcal{G}  & = &  opt - \sum_i \frac{opt}{1 + \vartheta_i(\lambda+1)}\cdot f_i \cdot \vartheta_i(\lambda+1)  \\
			& =    & opt \Bigg( 1 - \sum_i \frac{\vartheta_i(\lambda+1)}{1 + \vartheta_i(\lambda+1)}\cdot f_i   \Bigg) \\
			& \geq & opt \Bigg (1 - \sum_i f_i \cdot \frac{(\lambda +1)}{2+(\lambda +1)} \Bigg) = \,\,\,     opt \cdot \frac{2}{2+(\lambda+1)}
\end{eqnarray*}

\noindent
so that

\begin{eqnarray*}
\frac{sol}{opt} \,\,\,\geq\,\,\,  \frac{2}{2+(\lambda+1)} \in \Big[\frac{1}{2}, \frac{2}{3} \Big]
\end{eqnarray*}

\noindent
depending on the choice of $\lambda$.
\end{proof}

Now we  calculate how much each color bound $w_j$ can be violated.

\begin{lemma}
For the final solution $M$ returned by the algorithm and for every color class $C_j \in \mathcal{C}$ we have that $|M \cap E_j| \leq \frac{2w_j}{\lambda +1} +1$ for any choice of $\lambda \in [0,1]$.
\end{lemma}

\begin{proof}
Again we will provide an upper bound on the violation in terms of the chosen parameter $\lambda \in [0,1]$. As discussed earlier, at each application of a rounding step, we decrease the color bound $w_j$ of the color $C_j$ of the edge for which we performed the rounding step (step 5 of Algorithm \ref{bicriteria})  by $x_e + \lambda(1-x_e) \in [x_e,1]$. For every such rounding step, we include one such edge in our solution $M$. So the maximum number of edges of any particular color class $C_j$ we can include in $M$ (before $w_j$ truncates to zero) is

\begin{eqnarray}
\frac{w_j}{x_e + \lambda(1 - x_e)} \,\,\,=\,\,\, \frac{w_j}{x_e (1 - \lambda) + \lambda} \,\,\,\leq\,\,\, \frac{2w_j}{\lambda +1}
\end{eqnarray}
where the last inequality follows because of the fact that we perform rounding steps at edges with fractional value $\geq \frac{1}{2}$. The extra $+1$ term in the above formula comes from a possible relaxation step when $w_j \leq 1$ and there is at most one such step per color class.
\end{proof}

\begin{corollary}
For every $\lambda\in[0,1]$, there exists a $((\frac{2}{2+(\lambda+1)},0),(\frac{2}{\lambda +1},1))$ bi-criteria approximation algorithm for the uniform weight Bounded Color Matching problem in general graphs.
\end{corollary}

Observe that by selecting $\lambda \leq \frac{1}{2}$, we get an approximate solution strictly greater than $\frac{4opt}{7} \simeq 0.571$ (beating the $0.5$ integrality gap) in which each color bound $w_j$ is violated by \textit{at most} an additive $\frac{w_j}{3} + 1$ edges i.e. $M$ has a surplus of at most $\frac{w_j}{3} + 1$ of each color $C_j$.

\subsection{A 1/2-approximation algorithm for the uniform weight case}

So far, with the exception of the greedy 1/3 algorithm, we have presented bi-criteria algorithms that may potentially violate the color bounds $w_j$. Our main result is that by allowing moderate violation of these color bounds, we can ``beat" the integrality gap. The problem is these algorithms fail to achieve the 1/2 integrality gap \textit{without} violation (the two previous algorithms give 1/2 approximation on the objective function with at most one extra edge per color). In this subsection we will show how we can design a 1/2 approximation algorithm for the uniform weight Bounded Color Matching problem in general graphs, showing that the integrality gap is essentially 1/2.

The algorithm again uses the characterization provided by Lemma \ref{lemma}, with the only difference is that we \textit{do not} perform any relaxation step (such a step always result in a color bound violation). Instead, we perform only rounding steps which are summarized below (the rest of the algorithm is identical with the one of the previous subsection) i.e. the Relaxation step is replaced by another Rounding step:

\begin{itemize}
\item \textsf{if} $\exists v \in \mathcal{F}$ such that $\deg(v) =2$, \textsf{then} perform the usual Rounding step (see previous algorithm's step 5) with parameter $\lambda = 1$.
\item \textsf{else} $\exists C_j \in \mathcal{Q}$ such that

$$|support(x) \cap \chi_{E_j}| \leq w_j +1 \Rightarrow \exists e =\{ u,v \} \in C_j: x_e \geq \frac{w_j}{w_j+1}.$$

 Now, perform a Rounding step for that edge.
\end{itemize}

\noindent
The last step is done as follows: Round up $x_e$ to 1. Round down to zero all other edges adjacent to vertices $u$ and $v$ (the endpoints of e). Decrease $w_j$ by 1 and iterate.

We claim that this simple step result to a 1/2 approximation algorithm for the Bounded Color Matching problem without any violation:

\begin{lemma}
If instead or the Relaxation step (step 4. of Algorithm \ref{bicriteria}) we perform a Rounding step as described above, then Algorithm 4 is an 1/2-approximation algorithm for the BCM problem in uniform weighted graphs without any violation.
\end{lemma}

\begin{proof}
In order to prove the claim of the lemma we will distinguish between the two cases corresponding to the two different rounding steps. In each step, obviously, the gain that we have is 1 (we take one edge in our final solution).
We will show that in both cases the decrease of the optimal objective function value after the resolution of the LP is at most 2. In conclusion we will show that the gain over loss in each step is at least 1/2, and this will conclude the claim.

To this end, let $\mathrm{LP}(k)$ be the optimal objective function value at step $k$. At this step, we perform one of the two above rounding steps and resolve the new LP which will have optimal value $\mathrm{LP}(k+1)$. The main claim is that $\mathrm{LP}(k) - \mathrm{LP}(k+1) \leq 2$.


If the performed rounding step is done on an edge $e$ because of a vertex $v \in \mathcal{F}$ with degree 2, then the total decrease in the objective function value in the next iteration is at most

\begin{eqnarray*}
\sum_{e \in \delta(v)} x_e + \sum_{e' \in \delta(u) \setminus e} x_{e'} + (1-x_e) & \leq & \\
 1 + (1-x_e) + (1-x_e) & = & \\
 3 -2x_e & \leq & 2
\end{eqnarray*}

where the first term corresponds to the two edges adjacent to $v$ ($e$ will be rounded to one and the other to zero), the second term corresponds to the edges (excluding $e$) adjacent to $u$ that will be rounded to zero, and the third term corresponds to a potential reduce of some other edges of color $C_j$ such that $e \in C_j$, because of the color bound update step $w_j = w_j -1$. The reason for the last term is the following:  color $C_j$ of edge $e$ can be (almost) tight, thus if $x_{e}$ is close to 1/2 then this leaves us with a surplus of $1-x_e$ (close to 1/2) of the edges of color $C_j$. So, in the next iteration, the value of the rest of the edges of color $C_j$ will be reduced by at most $1-x_e$ but it can be the case that we cannot take advantage of this decrease to increase some other color class. For example, assume that $e$ is blue, $w_{blue} = 10$ and $x_{e} = 1/2$. Then $\sum_{e \in E_{blue}, e' \neq e} x_{e'} = w_{blue} -1/2 =9.5$ in the worst case. But the update step will reduce $w_{blue}$ by $1$ i.e. in the next iteration $w_{blue} = 9$ so the new LP solution will have to reduce the value of the blue edges by 1/2 (from 9.5 to 9).

Now we consider the case where the rounding step is done because of the presence of a tight color class $C_j \in \mathcal{Q}$ such that $|support(x) \cap \chi_{E_j}| = w_j +1$. In this case, we know that $\exists e \in C_j: x_e \geq \frac{w_j}{w_j+1}$. This edge will be rounded to one, and some other appropriate edges will be rounded to zero in such a way to preserve feasibility. Let $\{ u,v \} = e$ as before. Then when we round $x_e$ to one, we need to round all other edges adjacent to $u$ and $v$ to zero in order to have a feasible matching. In order to compute the total decrease $\Delta$ in the LP value by such a step, we first compute what we had before the performed rounding step. Thus, the total decrease of the LP value is at most

\begin{displaymath}
\Delta = \underbrace{(1~-~x_e)}_\text{{decrease on vertex $v$}} + \underbrace{(1~-~x_e)}_\text{{decrease on vertex $u$}} + \underbrace{(1~-~x_e)}_\text{{color bound update}} + ~~x_e  \,\,\,=\,\,\, 3 - 2x_e
\end{displaymath}

\noindent
i.e. the first two terms correspond to the loss due to rounding to zero the edges adjacent to $u$ and $v$, and the third term corresponds to loss due to color bound update (see the previous case for justification). The fourth term is simply the value of the edge $e$. Now, due to the fact that $\frac{w_j}{w_j+1} \geq \frac{1}{2}$, we have that

\begin{displaymath}
\Delta \,\,\,\leq\,\,\, 3 - 2x_e \,\,\,\leq\,\,\, 3 - 2 \cdot \frac{w_j}{w_j +1} \,\,\,\leq\,\,\, 3 -2 \cdot\frac{1}{2} = 2
\end{displaymath}

\noindent
Observe that $w_j$ remains  integral in such a case, because it is initially integer and in every step it is reduced exactly by one unit.

So, in conclusion, in each rounding step, the total accumulated loss is at most two units in the LP value, but the total gain is exactly one unit (we add one edge in $M$) proving the claim.
\end{proof}

\begin{theorem}
There exists an $\frac{1}{2}$ polynomial time approximation algorithm for the uniform weight Bounded Color Matching problem.
\end{theorem}



\section{Conclusions}										                                                    

In this work, we have presented bi-criteria approximation algorithms for the Bounded Color Matching problem (a.k.a. Restricted Matching problem) that achieve \textit{constant} approximation guarantee on both criteria of

\begin{enumerate}
\item maximizing the objective function value, and
\item minimizing the violation of the color constraints bounds.
\end{enumerate}

Our techniques were based on polyhedral characterizations of the natural linear program formulation of the problem (described by the $\mathcal{M}_c$ inequalities). This polyhedron has integrality gap $\frac{1}{2}$. We have presented an $\frac{1}{2}$-approximation algorithm for the uniform wight case and we have shown how, by allowing a slight violation of the color bound constraints, we can design approximation algorithms with better than $\frac{1}{2}$ guarantee (in the objective function value). Our proposed algorithm in fact is flexible enough to allow any desired guarantee within some given bounds, and provides a trade-off between the approximability of the objective function value and the violations of the color bounds.

Moreover, for the special case where $w_j =1$, $\forall j \in [k]$ we have shown how we can obtain an \textit{asymptotic} approximation guarantee (i.e. approximate the objective function to within arbitrary precision) but at the cost of violating the color bounds $w_j$ by at most an additive $\alpha -1$ for a given parameter $\alpha \in \mathbb{Z}^+$.

Given the limitation of the natural linear program formulation of the problem (captured by its integrality gap), it is natural to ask if there is another linear program formulation of the problem with better behavior. It is not obvious at all how such a linear program might look like (if it exists). But fortunately, there exists machinery from polyhedral theory, called ``lift-and-project" method, that allows us to strengthen a particular linear program by adding a set of valid inequalities. Many such lift and project methods have been proposed so far for example by Sherali and Adams \cite{DBLP:journals/siamdm/SheraliA90}, by Lov\'asz and Schrijver \cite{Lovász91conesof}, by Balas, Ceria and Cornu\'ejols \cite{DBLP:journals/mp/BalasCC93}, \cite{DBLP:conf/soda/BalasCC93} and by Lasserre \cite{DBLP:journals/siamjo/Lasserre02}, \cite{DBLP:journals/siamjo/Lasserre01}.

Let $P_0 = \{ x \in \{0,1\}^n:~Ax \leq \beta \}$, $A \in \mathbb{R}^{m \times n}, \beta \in \mathbb{R}^m$ be an initial polyhedron in the $n$-th dimensional space. All the previous techniques follow the same pattern: they operate in rounds, and in each round a specific set of linear  (Sherali-Adams and Balas, Ceria and Cornu\'ejols) or semi-definite (in the case of Lov\'asz and Schrijver and Lasserre) inequalities is added. Thus we obtain a \textit{hierarchy} of tighter formulations $K_n \subseteq K_{n-1} \subseteq \cdots \subseteq K_1$ of an initial relaxation of an integral polyhedron $I$ where $K_1$ is just the relaxation of $P_0$. The important features, common in all these hierarchies is that we can efficiently optimize any linear (or semi-definite) objective function over $K_t$ for any fixed $t$ and, moreover, after $n$ at most steps, we will arrive at an exact formulation of the convex hull of all the integral points if $I$, i.e., $K_n = P_0$.

All the previous hierarchies (except the Balas, Ceria and Cornu\'ejols) are placed in a common framework in the work of Monique Laurent \cite{DBLP:journals/mor/Laurent03} who proves, among many other things, that the Sherali-Adams hierarchy is incomparable than the Lov\'asz  Schrijver hierarchy but stronger that Lov\'asz  Schrijver \textit{with linear} lifting inequalities. Moreover, it is shown that the Lasserre hierarchy is stronger than any of the previous. It would be a very interesting research direction to investigate how the integrality gap changes after the application of any of these hierarchies. And, moreover, the possibility to obtain a tighter description of the convex hull of the integral points of the bounded matching polyhedron, leaves open the possibility of designing approximation algorithms with better performance guarantee.

\bigskip
\noindent
\textbf{Acknowledgements:} We would like to thank Christos Nomikos for introducing us to the problem and for many valuable and pleasant discussions. Moreover, we would like to sincerely thank the reviewers of the 2nd International Symposium on Combinatorial Optimization (ISCO '12) for carefully reading the preliminary version of this work and for many helpful suggestions and comments that improved the presentation of this work.

\bibliographystyle{acm}
\bibliography{references}
\end{document}